\documentclass[preprint,12pt]{elsarticle}

\usepackage{amsmath,amssymb,amsfonts}
\usepackage{algorithmic}
\usepackage{graphicx}
\usepackage{textcomp}
\usepackage[dvipsnames]{xcolor}
\usepackage{amsthm, tikz, hyperref}

\newtheorem{theorem}{Theorem}[section]
\newtheorem{corollary}[theorem]{Corollary}

\newtheorem{definition}[theorem]{Definition}
\newtheorem{proposition}[theorem]{Proposition}
\newtheorem{conjecture}[theorem]{Conjecture}

\journal{Information and Computation}

\begin{document}

\begin{frontmatter}



\title{Conjectural Decidability of the Skolem Problem}

\author{Florian Luca\fnref{label1,label2}}
\affiliation{
organization={Mathematics Division, Stellenbosch University},
city={Stellenbosch},
country={South Africa}
}
\author{Jo\"el Ouaknine\fnref{label2,label3,label4}}
\affiliation{
organization={Max Planck Institute for Software Systems},
addressline={Saarland Informatics Campus},
city={Saarbr\"ucken},
postcode={66123},
country={Germany}
}
\author{James Worrell\fnref{label5}}
\affiliation{%
organization={Department of Computer Science, Oxford University},
addressline={Wolfson Building, Parks Road},
city={Oxford},
postcode={OX1 3QD},
country={UK}
}
\fntext[label1]{Also affiliated with the Max Planck Institute for
  Software Systems, Germany.}
\fntext[label2]{Supported by ERC grant DynAMiCs (101167561).}
\fntext[label3]{Supported by DFG grant 389792660 as part of TRR 248.}
\fntext[label4]{Also affiliated with Keble College, Oxford as
  emmy.network Fellow.}
\fntext[label5]{Supported by EPSRC grant EP/X033813/1.}

\begin{abstract}
The Skolem Problem asks to determine whether a given integer linear
recurrence sequence (LRS) has a zero term. This problem, whose
decidability has been open for many decades, arises across a wide
range of topics in computer science, including loop termination,
formal languages, automata theory, and probabilistic model checking,
amongst many others.

In the present paper, we introduce a notion of ``large'' zeros of
(non-degenerate) linear recurrence sequences, i.e., zeros occurring at
an index larger than a double exponential of the magnitude of the
data defining the given LRS\@. We establish two main results. First,
we define an infinite set of prime numbers, termed ``good'', having
density one amongst all prime numbers, with the following property:
for any large zero of a given LRS, there is an interval around the
large zero together with an upper bound on the number of good primes
possibly present in that interval. The bound in question is much lower
than one would expect if good primes were distributed similarly as
ordinary prime numbers, as per the Cram\'er model in number theory. We
therefore conclude, conditionally on a strengthening of the classical
Cram\'er conjecture, that large zeros do not exist, which would entail
decidability of the Skolem Problem.  Second, we show unconditionally
that large zeros are very sparse: the set of positive integers that
can possibly arise as large zeros of some LRS has null density. This
in turn immediately yields a Universal Skolem Set of density one,
answering a question left open in the literature.
\end{abstract}




\end{frontmatter}

\section{Introduction}
An (integer) linear recurrence sequence (LRS) $\langle u_n\rangle_{n=0}^\infty$
is a sequence of integers satisfying a recurrence of the form
\begin{equation}
\label{eq:1}
u_{n+k}=a_1 u_{n+k-1}+\cdots+a_k u_n 
\end{equation}
where the coefficients $a_1,\ldots,a_k$ are integers.  The celebrated
theorem of Skolem, Mahler, and Lech~\cite{Sko34,Mah35,Lec53}
describes the set of zero terms of such a recurrence:
\begin{theorem}
  \label{thm:SML}
  Given an integer linear recurrence sequence $\langle u_n\rangle_{n=0}^\infty$, the
  set $\{n\in\mathbb{N}:u_n=0\}$ is a union of finitely many
  arithmetic progressions together with a finite set.
\end{theorem}

The statement of Thm.~\ref{thm:SML} can be refined by considering
the notion of \emph{non-degeneracy} of an LRS\@.  An LRS is
non-degenerate if in its minimal recurrence no quotient of two
distinct roots of the characteristic polynomial is a root of
unity.\footnote{For basic definitions, facts, and properties concerning linear
  recurrence sequences, we refer the reader to standard texts such as
  \cite[Chaps.~1 and 2]{BOOK}, \cite[Chap.~4]{TCT}, or \cite[Chap.~4]{Stanley}.}
A
given LRS can be effectively decomposed as the interleaving of finitely many
non-degenerate sequences, some of which may be identically zero.  The
core of the Skolem-Mahler-Lech theorem is the fact that a non-zero
non-degenerate linear recurrence sequence has finitely many zero
terms.  Unfortunately, all known proofs of this last result are
ineffective: it is not known how to compute the finite set of zeros of
a given non-degenerate linear recurrence sequence.  It is readily seen
that existence of a procedure to do so is equivalent to the existence
of a procedure to determine whether an arbitrary given LRS has a zero
term; the latter is known as the Skolem Problem. We
refer to~\cite[Chap.~6]{Berstel2010NoncommutativeRS}
and~\cite[Chap.~X]{Tao08} for
  expository accounts of the Skolem-Mahler-Lech theorem and discussion
  of the ineffectiveness of known proofs.

  In computer science, the Skolem Problem lies at the heart of key
decision problems in formal power series~\cite{RS94,BRS06}, stochastic
model checking~\cite{PiribauerB20}, control theory~\cite{BlondelT00,
  FijalkowOPP019}, and loop termination~\cite{OWSiglog15}.  The
problem is also closely related to membership problems on commutative
matrix groups and semigroups, as considered in~\cite{CLZ00,KL86}.  We
note that in several of the above-mentioned citations, the Skolem
Problem is used as a reference benchmark to establish hardness of
other open decision problems.

Decidability of the Skolem Problem is known only for certain special
cases, based on the relative order of the absolute values of the
characteristic roots.  Say that a characteristic root $\lambda$ is
\emph{dominant} if its absolute value is maximal among all the
characteristic roots.  Decidability is known in case there are at most
$3$ dominant characteristic roots, and also for recurrences of order
at most $4$~\cite{MST84,Ver85}.  However for LRS of order $5$ it is
not currently known how to decide the Skolem Problem.  For a (highly
restricted) subclass of LRS, the paper~\cite{AkshayBMV020} obtains
nearly matching complexity lower and upper bounds for the problem.

Some recent lines of research have succeeded in establishing
conditional decidability of the Skolem Problem for simple LRS (i.e.,
LRS none of whose characteristic roots are repeated), assuming certain
classical number-theoretic
conjectures~\cite{LLNOP022,BLNOPW22}. Nevertheless, to the best of our
knowledge, no putative algorithm has to date been proposed to solve
the Skolem Problem in full generality.

A different approach was initiated in~\cite{LOW21,LOW22,LMNOW} via the
notion of \emph{Universal Skolem Sets}. An infinite, recursive set
$\mathcal{S} \subseteq \mathbb{N}$ is a Universal Skolem Set if there
is some algorithm which, given any LRS, determines whether or not the
LRS has a zero in $\mathcal{S}$. Decidability of the Skolem Problem is
then of course equivalent to the assertion that $\mathbb{N}$ is itself
a Universal Skolem Set. The authors of~\cite{LOW21} succeded in
exhibiting a \emph{sparse} Universal Skolem Set, i.e., a set having
null density, and left open the question of whether Universal Skolem
Sets of strictly positive density, or even density one, could be
constructed (the interest in high-density Universal Skolem Sets being that they approximate
$\mathbb{N}$ more closely). The question was partially answered
in~\cite{LOW22}, which presented a positive-density Universal Skolem Set,
albeit restricted to simple LRS, and in~\cite{LMNOW}, which exhibited a
Universal Skolem Set of strictly positive density, and even
established density $1$ subject to the Bateman-Horn conjecture in
number theory.

In this paper we propose an explicit bound for the largest zero of a
non-degenerate LRS in terms of the data describing the LRS\@.  We call
zeros that exceed this bound \emph{large zeros} of the LRS\@.
Evidently, decidability of the Skolem Problem would follow from a
proof that large zeros do not exist.  Using deep upper bounds on the
cardinality of the set of zeros of non-degenerate algebraic LRS due to Amoroso
and Viada, we show that the set of positive integers arising as
large zeros of some non-degenerate LRS has null density, which in turn
yields a Universal Skolem Set of unconditional density one. We present
this result in Sec.~\ref{sec:sparse}.

While a proof that large zeros do not exist currently seems well out
of reach, we give a heuristic argument as to why this should
nevertheless be expected.  This argument is based on an analogue of
the well-known Cram\'{e}r conjecture on gaps between consecutive
primes.  This conjecture, originally formulated by Cram\'er in
1936~\cite{cramer1936} and subsequently refined by various number
theorists into its present form, asserts that, for some constant
$\kappa>1$, for every prime $p$ the distance to the next largest prime
is at most $\kappa (\log p)^2$.  The conjecture is based on the
heuristic that the sequence of prime numbers behaves similarly to a
Poisson-like random process in which the probability of a number $x$
being prime is $1/\log x$.  The largest observed prime gap is approximately
$0.9206 (\log p)^2$ (involving a search over all primes up to $4 \cdot
10^{18}$)~\cite{OHP14}, however the best known
upper bound on prime gaps is $O(p^{0.525})$, due to Baker, Harman, and
Pintz~\cite{BakerHarmanPintz2001}, which is far from Cram\'{e}r's
conjectured bound. Cram\'er himself proved that, under the Riemann
hypothesis, prime gaps are bounded above by
$O(p^{0.5}\log p)$~\cite{cramer1936}.  On the other hand, the best
known lower bound on largest prime gaps is
$\Omega\left( \frac{\log p \log \log p \log \log \log \log p}{\log
    \log \log p} \right)$, due to Ford, Green, Konyagin, Maynard, and Tao~\cite{ford18},
which is some way from the conjectured upper bound. We refer to~\cite{Granville1995} for
a discussion of Cram\'{e}r's conjecture and its refinements.

Here we define a subset of so-called \emph{good} primes based on
divisibility properties of LRS\@.  We show that the set of good primes
has density one in the set of all primes, or in other words that,
asymptotically speaking, almost all primes are good primes.  We
further show that if the Cram\'{e}r conjecture applies also to gaps
between consecutive good primes, then large zeros of LRS cannot
exist. The proof of the latter result 
proceeds by establishing an upper bound on the number of good primes
in the neighbourhood of a large zero that violates the conjectured
upper bound on gaps between good primes.  In other words, if good
primes are distributed according to Cram\'{e}r's heuristic then large
zeros cannot exist and the Skolem Problem is decidable.
  
\section{Background}
  We will need some basic notions concerning algebraic numbers.  All
  material can be found in~\cite{FrohlichTaylor1993}.  Recall that a \emph{number field}
  $\mathbb{K}$ is a subfield of $\mathbb{C}$ that is finite
  dimensional as a vector space over $\mathbb{Q}$.  We assume that
  $\mathbb K$ is a Galois extension of $\mathbb Q$, that is, it arises
  as the splitting field of a polynomial with integer coefficients.
  All elements of $\mathbb K$ are algebraic over $\mathbb Q$, that is,
  they arise as roots of polynomials with integer coefficients.  Those
  elements that arise more specifically as roots of monic polynomials
  with integer coefficients are called \emph{algebraic integers}.  The
  algebraic integers in $\mathbb{K}$ form a subring, denoted
  $\mathcal{O}_{\mathbb{K}}$.

For a number field $\mathbb{K}$, we denote by
$\mathrm{Gal}(\mathbb{K}/\mathbb{Q})$ the group of field automorphisms
of $\mathbb{K}$.  Given $\alpha \in \mathbb{K}$, the \emph{norm} of
$\alpha$ is defined by
\begin{gather*}
\mathcal{N}_{\mathbb{K}/\mathbb{Q}}(\alpha) = \!\!\!\!\! \prod_{\sigma \in 
  \mathrm{Gal}(\mathbb{K}/\mathbb{Q})} \!\!\! \!\!\! \sigma(\alpha) \, .
\end{gather*}
The norm $\mathcal{N}_{\mathbb{K}/\mathbb{Q}}(\alpha)$ is rational for all 
$\alpha\in\mathbb{K}$; moreover $\mathcal{N}_{\mathbb{K}/\mathbb{Q}}(\alpha) =
0$ iff $\alpha =0$, and
$\mathcal{N}_{\mathbb{K}/\mathbb{Q}}(\alpha)$ is 
an integer if $\alpha \in \mathcal{O}_{\mathbb{K}}$.  Clearly we have
$|\mathcal{N}_{\mathbb{K}/\mathbb{Q}}(\alpha)| \leq M^{d_{\mathbb{K}}}$, where
$d_{\mathbb{K}}$ is the degree of $\mathbb{K}$ and
\begin{gather*}
  M  := \!\!\! \max_{\sigma \in \mathrm{Gal}(\mathbb{K}/\mathbb{Q})}
  |\sigma(\alpha)|
\end{gather*}
is the \emph{house} of $\alpha$.

We recall that every ideal in $\mathcal{O}_{\mathbb{K}}$ can be
written uniquely up to the order of its factors as the product of
prime ideals.  Given a rational prime $P \in \mathbb{Z}$, we say that
a prime ideal $\mathfrak{p}$ \emph{lies above} $P$ if $\mathfrak{p}$ is a
factor of $P\mathcal{O}_{\mathbb{K}}$.  In this case we have that
$P \mid \mathcal{N}_{\mathbb{K}/\mathbb{Q}}(\alpha)$ for all
$\alpha \in \mathfrak{p}$.

Let $\mathfrak{p}$ be a prime ideal of ${\mathcal O}_{\mathbb K}$
lying above $P\in\mathbb Z$.  Recall that the Frobenius automorphism
$\sigma\in {\rm Gal}({\mathbb K}/{\mathbb Q})$ corresponding
to $\mathfrak{p}$ is such that
$\sigma(\alpha) \equiv \alpha^P \bmod\mathfrak{p}$ for all
$\alpha \in \mathcal{O}_{\mathbb{K}}$; in fact it is the unique Galois
automorphism with this property. Note however that for the
Frobenius automorphism to be
well-defined on ${\mathbb K}$, it is necessary for $P$ to be
unramified, for which it suffices to check that $P$ does not
divide the discriminant of ${\mathbb K}$. 

\section{Large Zeros and Good Primes}
\label{sec:large-zeros}
For an LRS $\mathbf{u} = \langle u_n\rangle_{n=0}^\infty$ as
in~\eqref{eq:1}, define its \emph{height}\footnote{Note that we consider
  here the \emph{magnitude} of the numbers defining a given LRS (rather than their
  bit size as is more common in complexity theory). An alternative
  definition in terms of bit size would of course be possible, only
  requiring adjusting~\eqref{eq:large} appropriately.}
to be 
\[
H_{\mathbf{u}}:=\max\{k,|a_1|,\ldots,|a_k|,|u_0|,\ldots,|u_{k-1}|\}
\, .
\]

Note that there are only finitely many different LRS of any specified
height bound. In the rest of the paper, we shall focus exclusively on
LRS ``of sufficient height''; by this we implicitly postulate the
existence of some absolute constant $C$ and further assume that all
LRS under consideration have height in excess of $C$. In our
subsequent development, it would be straightforward to provide
explicit suitable numerical values for each of the lower bounds that
we require,\footnote{In fact, in all of the height-related asymptotic inequalities
  that we derive, positing a lower bound of $12$ on the height 
  would suffice.} however an explicit numerical value for $C$ would still depend
on the precise formulation of the strengthening of the
Cram\'er-Granville conjecture that we introduce in Sec.~\ref{sec:cramer} (see
Conjecture~\ref{conj:good}).

We say that $n$ is a zero of $\mathbf{u}$ if $u_n=0$, and 
we say that it is a  \emph{large zero} if the inequality 
\begin{equation}
\label{eq:large}
n<\exp \exp(10H_{\mathbf{u}}\log H_{\mathbf{u}})
\end{equation}
fails. As we argue later on, there are good reasons to expect that
\eqref{eq:large} holds for \emph{all} zeros of all sufficiently high non-degenerate LRS, which
in turn would establish decidability of the Skolem
Problem.\footnote{The expression in \eqref{eq:large} has
  of course been chosen in order for our mathematical argument to go
  through. In actual fact, it is plausible to expect that an
  expression involving a \emph{single}
  exponential would suffice: as far as we are aware, there is
  currently no known construction of a family of non-degenerate LRS
  having zeros at indices of magnitude merely singly exponential in
  the height of the LRS, as defined above.}

\subsection{Bad Primes and Good Primes}

In this section, let $\mathbf{u}$ be a fixed sufficiently high
non-degenerate LRS satisfying~\eqref{eq:1}, and let us write
$H:=H_{\mathbf{u}}$ (that is, we omit the explicit dependence on
$\mathbf{u}$).

We can express the general term $u_t$ of
$\mathbf{u}$ in exponential-polynomial form, i.e.,
\begin{equation}
  \label{eq:exp-poly-solve}
u_t=\sum_{i=1}^s Q_i(t) \alpha_i^t \, ,
\end{equation}
where $s \leq H$ and $\alpha_1,\ldots,\alpha_s$ are the roots of the characteristic
polynomial
\[ x^k - a_1x^{k-1} - \cdots - a_k \]
of $\mathbf u$ and $Q_1,\ldots,Q_s$ are
univariate polynomials. Note that all characteristic roots are
algebraic integers since the characteristic polynomial is monic and
comprises exclusively coefficients in $\mathbb{Z}$.
Recall that if $\alpha_i$ has multiplicity $\mu_i$ as a
characteristic root then $Q_i$ has degree at most
$\mu_i -1$. Let
${\mathbb K}:={\mathbb Q}(\alpha_1,\ldots,\alpha_s)$.  The
coefficients of each $Q_i$ are in $\mathbb{K}$ and can
straightforwardly be computed
from the initial values $u_0,\ldots,u_{k-1}$ of the sequence by
solving a system of $k$ linear equations, thanks to~\eqref{eq:exp-poly-solve}.
By Cramer's determinant
rule,\footnote{This rule is named after the 18th-century Genevan
  mathematician Gabriel Cramer, who is presumably unrelated to the
  20th-century Swedish mathematician Harald Cram\'er, whose work plays
  an important role in motivating the present article.}
each of the coefficients of $Q_i$ is the quotient of an
algebraic integer by the determinant $\Delta := \mathrm{det}(M)$ of
the matrix $M$ below:\footnote{The matrix
$M$ has $s$ blocks, one for each characteristic root. For $\ell \in \{1,\ldots,s\}$ the
  $\ell$-th block has dimension $k\times \mu_{\ell}$ and has $(i,j)$-th
  element $(i-1)^{(j-1)} \alpha_{\ell}^{(i-1)}$ for $i \in \{1,\ldots,k\}$
and  $j \in \{1,\ldots,\mu_\ell\}$.}
\[
\begin{bmatrix} 1 & \ldots & 0 & 1 & \ldots & 0 & 1 & \cdots\\
\alpha_1 & \ldots & \alpha_1 & \alpha_2 & \ldots & \alpha_{s-1} & \alpha_s & \ldots\\
\vdots & \ddots & \vdots & \vdots & \ddots & \vdots & \vdots & \ddots\\
\alpha_1^{k-1} \; & \ldots \; & (k-1)^{\mu_1-1} \alpha_1^{k-1} \; &
\alpha_2^{k-1} \; & \ldots \; & (k-1)^{\mu_s-1} \alpha_{s-1}^{k-1} \; &
\alpha_s^{k-1} &\ldots \end{bmatrix} \, .
\]
By the Cauchy root bound we have 
$|\alpha_i|\leq 1+H$ for $i\in\{1,\ldots,s\}$.  It follows that
the squared Euclidean norm of each column vector above is at most 
\begin{gather*}
{k(k-1)^{2(k-1)} (1+H)^{2k}}<k^{2k} (1+H)^{2k}\, .
\end{gather*}
Thus, by the Hadamard inequality,
\[|\Delta|^2<(k^{2k} (1+H)^{2k})^k=(k(1+ H))^{2k^2} \, .\]

The determinant $\Delta$ is in general, of course, a complex number. Note however that
any Galois automorphism $\sigma \in
\mathrm{Gal}(\mathbb{K}/\mathbb{Q})$ will permute the characteristic
roots, and thus when applied to $M$ will have the effect of
permuting its columns. As a result, for any such $\sigma$,
$\sigma(\Delta) = \pm \Delta$, and therefore the quantity $\Delta^2$
is stable under Galois automorphisms. We conclude that $\Delta^2$
must be a rational number, and since it is also by construction an
algebraic integer,\footnote{Note that every entry of $M$ is an
  algebraic integer.} we must have $\Delta^2 \in \mathbb{Z}$.

Let us now consider the LRS 
$\mathbf{v} := \Delta^2\mathbf{u}$, noting that $\mathbf{u}$ and
$\mathbf{v}$ share the same zeros. Writing
\[
v_t=\sum_{i=1}^s P_i(t) \alpha_i^t \, ,
\]
we observe that all the coefficients of each of the $P_i$ are
algebraic integers. We therefore have, for each $1 \leq i \leq s$,
\[
P_i(t) = \Delta^2 Q_i(t)=\sum_{j=0}^{\mu_i-1} c_{i,j} t^j \, .
\]
We wish to estimate the size of each $c_{i,j}\in
{\mathcal O}_{\mathbb K}$. From our earlier calculation via Cramer's
determinant rule, noting that $|u_0|,\ldots,|u_{k-1}|$ are all bounded
above by $H \leq 1+H$, and invoking the Hadamard inequality once more,
we conclude that the house of each $c_{i,j}$ is bounded above by
\begin{equation}
\label{eq:house}
  |\Delta| (k^k(1+H)^k)^k < (k(1+H))^{2k^2} < (1+H)^{4H^2}< H^{H^3}
  \, .
\end{equation}  

Let $\sigma\in \Sigma_s$ be any permutation of the first $s$ integers and let 
\[
\beta_i:=\alpha_{\sigma(i)}\quad {\text{\rm for}}\quad i=1,\ldots,s
\, .
\]
For some nonnegative integer $m$ consider the algebraic integer
\begin{equation}
\label{eq:21}
v_{m,\sigma}=\sum_{i=1}^s P_i(m)\beta_i \alpha_i^{m} \, .
\end{equation}

We need one last technical ingredient in order to define what it means for a prime number to be
\emph{bad}. Let $X > \exp \exp e$ be a power of $2$. We say that \emph{$\mathbf{u}$ is
small at level $X$} provided that
  \[H < \frac{\log \log X}{10\log \log \log X} \, .\]

\begin{definition} 
  We say that a prime $P > \exp \exp e$ is bad, if there exist $X$ a
  power of $2$ with $X/2 < P < X$, together with an LRS
  $\mathbf{u}$ which is small at level $X$,
  a permutation $\sigma\in \Sigma_s$, and an integer $m\in [0,X^{1/4}]$, such that 
\begin{itemize}
\item The algebraic integer $v_{m,\sigma}$ defined in \eqref{eq:21} above is non-zero, and 
\item $P$ is a prime factor of ${\mathcal N}_{{\mathbb K}/{\mathbb Q}} (v_{m,\sigma})$.
\end{itemize}
Let ${\mathcal P}_{\text{\rm bad}}(X)$ be the set of bad primes in
$[X/2,X]$, and 
${\mathcal P}_{\text{\rm bad}} := \bigcup {\mathcal P}_{\text{\rm
      bad}}(2^k)$,
where the union is taken over all positive integers $k$ such that 
$2^k>\exp \exp e$.
\end{definition}

\begin{proposition}
\label{prop:density}  
We have 
\[
\# {\mathcal P}_{\text{\rm bad}}(X)<X^{2/3}
\]
for all $X>X_0$, where $X_0$ is some effective absolute constant. 
\end{proposition}

\begin{proof}
In order to estimate the size of ${\mathcal P}_{\text{\rm bad}}(X)$,
we first need to find out: 
\begin{enumerate}
\item[\textbf{1}.] How many such expressions \eqref{eq:21} are there?
\item[\textbf{2}.] How large are they?
\end{enumerate}

For (\textbf{1}), let us count the number of distinct possible LRS of
size at most $H$. Such LRS have coefficients $a_1,\ldots,a_k$ and initial values
$u_0,\ldots,u_{k-1}$ all in $[-H,H]$, an interval containing at
most $2H+1<3H$ integers. Altogether for fixed $k$ there are at most
$(3H)^{2k}\le (3H)^{2H}$ $2k$-tuples, and summing up over $k$ we
derive an upper bound of $H(3H)^{2H}<H^{3H}$ distinct
possible LRS of size at most $H$.

This in turn is an upper bound on
the number of $s$-tuples $((Q_i ,\alpha_i))_{i=1}^s$. We must then
multiply this quantity with the number of possible permutations of the
characteristic roots,
which is at most $H!<H^H$. There are therefore
at most $H^{4H}$ linear recurrence sequences
$\mathbf{w} = \langle w_m \rangle_{m=0}^{\infty}$ whose $m$-th term is
given by
\[
w_m=\sum_{i=1}^s P_i(m)\beta_i\alpha_i^m\quad {\text{\rm
    for~all}}\quad m\ge 0 \, .
\]

This answers (\textbf{1}).
As for (\textbf{2}), recall that the coefficients of $P_i$ are of
absolute value at most 
$H^{H^3}$, as per~\eqref{eq:house}. $P_i(m)$ comprises at most $H$ monomials, the largest of which
is at most $m^H < X^H$, and the largest root has
magnitude at most
$1+H < 2H$.
Thus each individual term $w_m$ is of absolute value at most
\begin{multline*}
  H^{H^3+1} (2H) X^H (2H)^{X^{1/4}} = \\
  \exp\left((H^3+1)\log H + \log (2H) + H \log X + X^{1/4}\log (2H)
  \right)\\ <\exp(X^{0.26})
\end{multline*}
for $X>X_0$, since $H$ is tiny in comparison to $X$.
Hence the norm of the number shown in \eqref{eq:21} is of size at most 
\[
\exp(H! X^{0.26})<\exp(X^{0.27})\quad {\text{\rm for}}\quad X>X_0 \,
,
\]
since the degree of $\mathbb{K}$ is at most $H!$ (as $\mathbb{K}$ is
the splitting field of a polynomial of degree at most $H$).
Moreover, as noted
earlier, there are at most $H^{4H}$ such expressions.
Thus a bad prime $P$ divides an integer which is a product of such numbers and is of size at most 
\[
\exp(H^{4H} X^{0.27})<\exp(X^{0.28})\quad {\text{\rm for}}\quad
X>X_0 \, .
\]

Therefore the number of possible choices for $P$ is at most $X^{0.28}$. Since
the number of choices for $m$ is at most $X^{0.25}$, we conclude that,
for $X > X_0$,
the cardinality of $\mathcal{P}_{\text{\rm bad}}(X)$ is at most
\[
X^{0.25+0.28} < X^{2/3} \, ,
\]
as required.
\end{proof}

Finally, let us write $\mathcal{P} = \{p_1, p_2, \ldots \}$ to denote
the set of prime numbers, enumerated in increasing order, and let
$\mathcal{P}_{\mathrm{good}} := \mathcal{P} \setminus
\mathcal{P}_{\mathrm{bad}} = \{g_1, g_2, \ldots\}$ denote the subset of
\emph{good} primes, again enumerated in increasing order.  Note that,
by Prop.~\ref{prop:density} along with the prime number theorem, the
set of bad primes has null density amongst the prime numbers. This in
turn entails that good primes have density one amongst all prime
numbers.

\section{The Cram\'er Argument}
\label{sec:cramer}

In this section we present a heuristic argument supporting
the
assertion that large zeros of sufficiently high LRS do not exist.  The strategy is as
follows. Assuming that good
primes are distributed similarly as ordinary primes, then according to the
Cram\'{e}r model in number theory, one would expect that Cram\'{e}r's
conjecture on gaps between primes applies also to good primes. More
precisely, this conjecture postulates the existence of precise upper
bounds on the largest possible gap between consecutive primes, and is
predicated on the heuristic that the primes behave as a set of
randomly distributed integers with asymptotic density conforming to
the prime number theorem. However we show that
around any large zero of an LRS there is an interval and an upper
bound on the number of good primes in the interval that together contradict
the above Cram\'{e}r-type conjecture on gaps between good primes.  We
therefore surmise that large zeros do not exist.

Recall that $\mathcal{P} = \{p_1, p_2, \ldots \}$ denotes the set
of prime numbers enumerated in increasing order.

\begin{conjecture}[Cram\'er-Granville]
  \label{conj:cramer}
  For some $\kappa > 1$,
  \[
    \limsup_{j \rightarrow \infty} \frac{p_{j+1} - p_j}{(\log p_j)^2} =
    \kappa \, .
\]    
\end{conjecture}

Cram\'er initially suggested that the constant $\kappa$ in
Conjecture~\ref{conj:cramer} might be 1~\cite{cramer1936}, but several
decades later, building on substantial developments in the field,
Granville produced evidence that $\kappa \geq 2e^{-\gamma} \approx
1.1229\ldots$, where $\gamma$ is the Euler–Mascheroni
constant~\cite{Granville1995}. There is in any event considerable
computational evidence in support of the Cram\'er-Granville
conjecture~\cite{Odlyzko1999,nicely99}. 

As noted earlier, thanks to Prop.~\ref{prop:density} and the
prime number theorem, good primes have density one amongst all
prime numbers:
\[ \lim_{X \rightarrow \infty}
  \frac{\# \left(\mathcal{P}_{\mathrm{good}} \cap [0,X]\right)}{\#\left(\mathcal{P} \cap [0,X]\right)}
  = 1 \, .
  \]

  In other words, asymptotically speaking, almost all primes
  are good primes. Accordingly, it seems reasonable to suppose that
  good primes should behave similarly to ordinary primes, or at least
  should exhibit similar ``statistical'' properties. We therefore
  formulate the following strengthening of the Cram\'er-Granville conjecture:

\begin{conjecture}
  \label{conj:good}
  For some $\eta > 1$,
  \[
    \limsup_{j \rightarrow \infty} \frac{g_{j+1} - g_j}{(\log g_j)^2} =
    \eta \, .
\]    
\end{conjecture}

We now have the following result.

\begin{theorem}
  \label{thm:no-large-zeros}
Conjecture~\ref{conj:good} implies that large zeros of sufficiently
high LRS do not exist. More precisely, assuming
Conjecture~\ref{conj:good}, there exists an absolute constant $C>0$
such that, for all non-degenerate LRS $\mathbf{u}$ of height
$H_\mathbf{u}$ at least $C$, whenever $u_n = 0$ then
$n < \exp \exp(10 H_\mathbf{u} \log H_\mathbf{u})$.
\end{theorem}

\begin{proof}

Conjecture~\ref{conj:good} can be reformulated as follows: there
exist $\eta > 1$ and $n_0 \in \mathbb{N}$ such that, for all $n \geq n_0$, the interval
\[
[n-\eta(\log n)^2, n]
\]
always contains some good prime. In turn, this implies that the
interval $[n-0.5n^{1/4},n]$ must contain at least
$n^{1/4}/(2\eta(\log n)^2)$
distinct good primes for $n$ sufficiently large.

Suppose now that there is some sufficiently high LRS $\mathbf{u}$
having large zero $u_n = 0$. By definition, this means that
\begin{equation}
\label{eq:large-zero-calc}  
  n \geq \exp \exp (10 H_{\mathbf{u}} \log H_{\mathbf{u}}) \, .
\end{equation}
Let $X$ be the power of $2$ such that $X/2 \leq n < X$. We
first claim that $\mathbf{u}$ is small at level $Y := X/2$, i.e., that the
inequality
\begin{equation}
  \label{eq:salx}
  H_{\mathbf{u}} < \frac{\log \log Y}{10\log \log \log Y}
\end{equation}  
holds.

Suppose for a contradiction that Eq.~\eqref{eq:salx} fails. Then
\[  \log H_{\mathbf{u}} \geq \log \log \log Y
 - \log(10 \log \log \log Y) \geq \frac{3}{4}\log \log \log Y \, , \]
where the second inequality follows from the assumption that $Y$ is
sufficiently large (since $\mathbf{u}$ is assumed to be sufficienctly
high). Combining with the negation of~\eqref{eq:salx}, we get
\[ 10 H_{\mathbf{u}} \log H_{\mathbf{u}} \geq 10 \left(
  \frac{\log \log Y}{10 \log \log \log Y} \right) \frac{3}{4} \log
\log \log Y = \frac{3}{2}\log \log Y
 \, , \]
i.e., $\exp \exp (10 H_{\mathbf{u}} \log H_{\mathbf{u}}) \geq Y^{(\log Y)^{0.5}}$,
whence (from~\eqref{eq:large-zero-calc}) we conclude that
$n \geq  Y^{(\log Y)^{0.5}}$, contradicting (for sufficiently large
$Y$) the fact that $n < X = 2Y$ and thereby establishing the claim.

We note that the smallness of $\mathbf{u}$ at level $X/2$ immediately
also entails its smallness at level $X$.

Next, write $n=P+m$, where
$P\in [n-0.5n^{1/4},n]$ is a good prime and
$0\leq m< n^{1/4}$.  As in the previous
section, let $\alpha_1,\ldots,\alpha_s$ be the
characteristic roots of $\mathbf{u}$, put $\mathbb{K} :=
\mathbb{Q}(\alpha_1,\ldots,\alpha_s)$, and let $\Delta^2$ be
the smallest positive integer such that, writing
$\mathbf{v} := \Delta^2 \mathbf{u}$, every term of $v_t$ of $\mathbf{v}$ has a
representation as an exponential polynomial
\[
  v_t = \sum_{i=1}^s P_i(t)\alpha_i^{t}
  \]
  in which all polynomials $P_i$ have algebraic-integer coefficients.
We let $k$ stand for the order of $\mathbf{u}$ and abbreviate
$H_{\mathbf{u}}$ as $H$ for the remainder of the proof.
  
Since $u_n= v_n = 0$, we get
\[
0=\sum_{i=1}^s P_i(P+m)\alpha_i^{P+m} \, .
\]
We now reduce the above equation modulo $\mathfrak{p}$, where
$\mathfrak{p}$ is some prime ideal of ${\mathcal O}_{\mathbb K}$
dividing $P$, from which we deduce that $P$ divides
\begin{equation}
\label{eq:3}
{\mathcal N}_{{\mathbb K}/{\mathbb Q}}\left(\sum_{i=1}^s
  P_i(m)\beta_i\alpha_i^m\right) \, ,
\end{equation}
where each $\beta_i=\sigma(\alpha_i)$ is obtained from applying the
Frobenius automorphism induced by $\mathfrak{p}$ in ${\mathbb K}$ to
$\alpha_i$. Recall that it is necessary for this automorphism to be
well defined that $P$ not divide the discriminant of ${\mathbb K}$.
Note however that the discriminant  of $\mathbb K$ is a positive integer
bounded above by the absolute value of the discriminant of the
characteristic polynomial of $\mathbf{u}$, namely
$\prod_{ i \neq j} |\alpha_i-\alpha_j|$, and this last quantity is
itself at most $(2+H)^{s(s-1)} \leq H^{H^3}$ thanks to the Cauchy root
bound $|\alpha_i| \leq 1+H$. Since $\mathbf{u}$ is small at level
$X/2$ and $X/2 \leq n$, we have $H < (\log \log n) / (10 \log \log \log
n)$, and hence
\[
H^{H^3} < \exp\left(\left(\frac{\log\log n}{10\log\log\log n}\right)^3 \log\log\log n\right)<\exp((\log\log
n)^3)<n/2 \, . \]
Since $P$ is larger than $n/2$ for $n$ sufficiently large, we conclude
that $P$ cannot divide the discriminant  of $\mathbb K$, as required.

Observe that $P$ must lie in one of two intervals,
namely $[X/4,X/2]$ or $[X/2,X]$. Since we know that $\mathbf{u}$ is
small at both levels $X/2$ and $X$, it follows that
expression~\eqref{eq:3} cannot be non-zero, otherwise $P$, by definition,
would be a bad prime. Expression~\eqref{eq:3} is therefore equal to zero.

Let us count how many expressions of the form~\eqref{eq:3} can
vanish. More precisely, consider the (complex-valued) LRS
$\mathbf{w} = \langle w_j \rangle_{j=0}^{\infty}$ whose $j$-th term is
given by
\[
w_j=\sum_{i=1}^s P_i(j)\beta_i\alpha_i^j\quad {\text{\rm
    for~all}}\quad j\ge 0 \, ,
\]
and whose order is at most $k$. Amoroso and Viada~\cite{AV11}
prove that the number of distinct positive integers $m$ such that $w_m = 0$
is at most
\begin{multline*}
(8k^k)^{8(k^k)^6} =  \exp(8k^{6k}\log (8k^k))\le
\exp(8H^{6H+1}\log(8H)) 
 < \\ \exp(16 H^{6H+1}\log H)<\exp(H^{7H})=\exp \exp (7H\log H) \, .
\end{multline*}

Of course, given $\mathbf{u}$, the $s$-tuple $(\beta_1,\ldots,\beta_s)$
can be chosen in at most $s!<H^H$ ways.
Thus the total number of possible zeros for expressions of the form~\eqref{eq:3} is
at most
\begin{multline*}
  H^H\exp \exp (7H\log H)<\exp \exp(7H\log H+\log H+\log\log H) \\ 
 <  \exp \exp (8H\log H) \, .
\end{multline*}
Since
distinct choices of $P$ give rise to distinct such
zeros,\footnote{Recall that $n = P + m$, and thus distinct choices of
  $P$ entail distinct values of $m$.} and (as
noted earlier) there are at least
$n^{1/4}/(2\eta(\log n)^2)$ possible choices for $P$,
we conclude that
\[
\frac{n^{1/4}}{2\eta(\log n)^2} < \exp \exp (8H\log H) \, ,
\]
or equivalently
\[ \frac{n^{1/4}}{(\log n^{1/4})^2} < 32\eta \exp \exp(8H \log H) \,
  . \]

Observe that, for any fixed $D>0$, the inequality $x/(\log x)^2 < y$
implies that $x < y^2/D$, provided that $x$ is sufficiently
large. Applying this to the above inequality with $x=n^{1/4}$ and
$D = 32\eta$, we derive the upper bound
\[ n^{1/4} < \exp(2\exp(8H \log H)) \, , \]
whence
\[ n < \exp(8\exp(8H \log H)) < \exp \exp (9H \log H) \, , \]
contradicting Eq.~\eqref{eq:large-zero-calc} to the effect that $n$ is
a large zero of $\mathbf{u}$. We conclude that large zeros do not
exist, as required.
\end{proof}

\section{Large Zeros Are Unconditionally Sparse}
\label{sec:sparse}

In this section, we prove unconditionally that large zeros are sparse,
i.e., have null density amongst the positive integers.

To this end, let
\begin{align*}
{\mathcal L} := \{n \in \mathbb{N} : {}&{\text{\rm there exists a
                                         non-degenerate LRS \textbf{u}
                                         such that}}\\ &
                                         u_n=0~{\text{\rm
    and}}~n\geq \exp \exp(10H_{\mathbf{u}}\log H_{\mathbf{u}})\} \, . 
\end{align*} 
Thus ${\mathcal L}$ is the set of large zeros of \emph{some}
non-degenerate LRS, without any height restrictions.

\begin{theorem}
  \label{thm:density-0}
  The set ${\mathcal L}$ has null density.
  In fact, writing ${\mathcal L}(X)={\mathcal L} \mathop{\cap} [1,X]$, the inequality 
\[
\#{\mathcal L}(X)=O\left(\frac{X}{(\log X)^B}\right) 
\]
holds with any constant $B>0$. 
\end{theorem}   
\begin{proof}
We let $X$ be large, and aim to count the $n$ in $\mathcal{L}$ that
lie in $[X/2,X]$. Jia~\cite{Jia96} proved that the set of $n\in
[X/2,X]$ such that the interval $I_n :=[n-n^{1/19}, n]$
contains fewer than $X^{1/19}/(\log X)^2$ primes is of counting
function  $O(X/(\log X)^B)$ with any $B>0$. Let us therefore
consider $n \in [X/2,X]$ such that
$I_n$ contains at least $X^{1/19}/(\log X)^2$ primes.
Write $n=P+m$, where $P\in I_n$ and $m<X^{1/19}$. Let ${\bf u}$
be an LRS having $n$ as a large zero and whose recurrence is given by 
\eqref{eq:1}. Using the same notation and reasoning as in the proof of
Thm.~\ref{thm:no-large-zeros},
we have, for sufficiently large $X$, that $P$ must divide
\begin{equation}
\label{eq:4}
{\mathcal N}_{{\mathbb K}/{\mathbb Q}}\left(\sum_{i=1}^s
  P_i(m)\beta_i\alpha_i^m\right) \, .
\end{equation}

If the above expression is non-zero, then $P$ is a bad prime, which in
turn means that $n$ is within $O(X^{1/19})$ of a bad prime. However, as
shown in Prop.~\ref{prop:density}, the number of bad primes below $X$
is $O(X^{2/3})$. Hence the total number of such $n\in [X/2,X]$ is 
$O(X^{2/3+1/19})=O(X/(\log X)^B)$ with any $B > 0$.

Finally, assume that expression~\eqref{eq:4} is zero, in which
case
\[
\sum_{i=1}^s P_i(m)\beta_i\alpha_i^m=0 \, .
\]
Once again, as detailed in the proof of Thm.~\ref{thm:no-large-zeros},
the Amoroso-Viada bounds imply that the corresponding number of
possible zeros for expressions of the form above is at most
$\exp \exp (8 H_{\mathbf{u}} \log H_{\mathbf{u}})$. Since distinct
choices of $P$ give rise to distinct such zeros, and since
$I_n$ contains at least $X^{1/19}/(\log X)^2$ primes, we deduce that
\[ \frac{X^{1/19}}{(\log X)^2} < \exp \exp (8 H_{\mathbf{u}} \log
  H_{\mathbf{u}}) \, . \]
Noting that $n \leq X$, we have, for sufficiently large $n$, the following string of
inequalities:
\begin{multline*}
  n < \left( \frac{n^{1/19}}{(\log n)^2}\right)^{20} \leq
  \left( \frac{X^{1/19}}{(\log X)^2}\right)^{20} < 
  (\exp \exp (8 H_{\mathbf{u}} \log
  H_{\mathbf{u}}))^{20} \\ < \exp \exp (9 H_{\mathbf{u}} \log
  H_{\mathbf{u}}) \, ,
\end{multline*}  
or in other words that $n \notin \mathcal{L}$.

Putting everything together, we conclude that, for any $B > 0$, if $X$
is sufficiently large then the number of large zeros in $[X/2,X]$ is
$O(X/(\log X)^B)$, from which one immediately deduces the statement of
the theorem.
\end{proof}

\begin{corollary}
The set $\mathcal{S} := \mathbb{N} \setminus \mathcal{L}$ is a Universal Skolem Set of
density one.
\end{corollary}

\begin{proof}
  It is clear that the set $\mathcal{L}$ is recursive, and hence that
  $\mathcal{S}$ is recursive as well.

  Density one follows from Thm.~\ref{thm:density-0}, and universality
  follows from the fact that $\mathcal{S}$, by definition, doesn't
  contain any large zeros. Thus given any non-degenerate LRS
  $\mathbf{u}$ of size $H_{\mathbf{u}}$, its only possible zeros in
  $\mathcal{S}$ can only lie in the interval
  $[0, \exp \exp (10 H_{\mathbf{u}} \log
  H_{\mathbf{u}})]$, which can readily be checked.
\end{proof}

\section{Concluding Remarks}

Thanks to Thm.~\ref{thm:no-large-zeros}, Conjecture~\ref{conj:good}
implies the \emph{existence} of an algorithm to solve the Skolem
Problem, as follows. Recall that Thm.~\ref{thm:no-large-zeros} asserts
the existence of an absolute constant $C$ such that
any non-degenerate LRS $\mathbf{u}$ of height $H_{\mathbf{u}} \geq C$
has no zeros at index larger than $\exp \exp (10 H_{\mathbf{u}} \log
  H_{\mathbf{u}})$. On the other hand, there are only finitely many
  LRS of height at most $C$; therefore there exists \emph{some}
  algorithm (call it an oracle) that correctly determines, for each such $C$-bounded LRS,
  whether or not it has a zero.

  Now given an LRS $\mathbf{u}$, first decompose $\mathbf{u}$ into
  finitely many non-degenerate LRS, and check that none of these is
  identically zero. Next, if any of these LRS has height below $C$,
  invoke the aforementioned oracle to determine whether it has a
  zero. Finally, assuming that no zeros have yet been found, for each
  remaining LRS $\mathbf{v}$, check whether $\mathbf{v}$ has a zero in the
  interval $[0, \exp \exp (10 H_{\mathbf{v}} \log H_{\mathbf{v}})]$.

  Of course, even if Conjecture~\ref{conj:good} were to be established
  and $C$ explicitly known, and setting aside the question of how to
  obtain the oracle, the above algorithm unfortunately remains impractical, since
  the magnitudes involved are much too large to envisage any reasonable implementation.

\bibliographystyle{elsarticle-num} 
\bibliography{cramer.bib}

@article{OHP14,
  author    = {T. {Oliveira e Silva} and S. Herzog and S. Pardi},
  title     = {Empirical verification of the even {G}oldbach
conjecture and computation of prime gaps up to $4 \cdot 10^{18}$},
  journal   = {Math. Comp.},
  volume    = {83},
  number    = {288},
  pages     = {2033--2060},
  year      = {2014},
  publisher = {American Mathematical Society}
}

@article{LMNOW,
  author       = {F. Luca and
                  J. Maynard and
                  A. Noubissie and
                  J. Ouaknine and
                  J. Worrell},
  title        = {Skolem Meets Bateman-Horn},
  journal      = {CoRR},
  volume       = {abs/2308.01152},
  year         = {2023}
}

@book{FrohlichTaylor1993,
  author    = {A. Fröhlich and M. J. Taylor},
  title     = {Algebraic Number Theory},
  series    = {Cambridge Studies in Advanced Mathematics},
  volume    = {27},
  year      = {1993},
  publisher = {Cambridge University Press}
}

@article{Odlyzko1999,
  author    = {A. Odlyzko and M. Rubinstein and M. Wolf},
  title     = {Jumping Champions},
  journal   = {Experimental Mathematics},
  volume    = {8},
  number    = {2},
  pages     = {107--118},
  year      = {1999},
  publisher = {Taylor & Francis}
}

@article{Jia96,
author = {C. Jia},
title = {Almost All Short Intervals Containing Prime Numbers},
journal = {Acta Arith.},
volume = {76},
number = {1},
pages = {21--84},
year = {1996},
publisher={Instytut Matematyczny Polskiej Akademii Nauk}
}

@article{AV11,
author = {F. Amoroso and E. Viada},
title = {On the Zeros of Linear Recurrence Sequences},
journal = {Acta Arith.},
volume = {147},
number = {4},
pages = {387--396},
year = {2011},
publisher={Instytut Matematyczny Polskiej Akademii Nauk}
}

@article{Granville1995,
  author    = {A. Granville},
  title     = {Harald {C}ramér and the Distribution of Prime Numbers},
  journal   = {Scandinavian Actuarial Journal},
  volume    = {1995},
  number    = {1},
  pages     = {12--28},
  year      = {1995},
  publisher = {Taylor & Francis}
}

@article{BakerHarmanPintz2001,
  author    = {R. C. Baker and G. Harman and J. Pintz},
  title     = {The Difference Between Consecutive Primes, II},
  journal   = {Proceedings of the London Mathematical Society},
  volume    = {83},
  number    = {3},
  pages     = {532--562},
  year      = {2001},
  publisher = {Oxford University Press}
}

@Book{BOOK,
  author =       {G. Everest and A. van der Poorten and
                  I. Shparlinski and T. Ward},
  title =        {Recurrence Sequences},
  publisher =    {American Mathematical Society},
  year =         2003}

@article{Stanley,
  title={Enumerative Combinatorics},
  volume = {Volume 1, 2nd Edition},
  author={Stanley, R. P.},
  journal={Cambridge studies in advanced mathematics},
  year={2011}
}

@book{TCT,
  author       = {M. Kauers and
                  P. Paule},
  title        = {The Concrete Tetrahedron --- Symbolic Sums, Recurrence Equations, Generating
                  Functions, Asymptotic Estimates},
  series       = {Texts {\&} Monographs in Symbolic Computation},
  publisher    = {Springer},
  year         = {2011}
}

@InProceedings{Sko34,
  author =       {T. Skolem},
  title =        {Ein {V}erfahren zur {B}ehandlung gewisser
exponentialer {G}leichungen},
  booktitle = {Comptes rendus du congr\`es des math\'ematiciens
scandinaves},
  year =         1934}

@Article{Mah35,
  author =       {K. Mahler},
  title =        {Eine arithmetische {E}igenschaft der
{T}aylor {K}oeffizienten rationaler {F}unktionen},
  journal =      {Proc. Akad. Wet. Amsterdam},
  year =         1935,
  volume =       38}

@Article{Lec53,
  author =       {C. Lech},
  title =        {A note on recurring series},
  journal =      {Ark. Mat.},
  year =         1953,
  volume =       2}

@book{Tao08,
  author  = {T. Tao},
  title = {Structure and randomness: pages from year one of a mathematical blog},
  publisher = {American Mathematical Society},
  year = {2008}
}

@Article{MST84,
  author =       {M. Mignotte and T. Shorey and R. Tijdeman},
  title =        {The distance between terms of an algebraic
                  recurrence sequence},
  journal =      {J. f\"ur die reine und angewandte Math.},
  year =         1984,
  volume =       349}

@Article{Ver85,
  author =       {N. K. Vereshchagin},
  title =        {The problem of appearance of a zero in a linear
                  recurrence sequence (in {R}ussian)},
  journal =      {Mat. Zametki},
  year =         1985,
  volume =       38,
  number =       2}

@inproceedings{AkshayBMV020,
  author       = {S. Akshay and
                  N. Balaji and
                  A. Murhekar and
                  R. Varma and
                  N. Vyas},
  title        = {Near-Optimal Complexity Bounds for Fragments of the {S}kolem Problem},
  booktitle    = {{STACS}},
  series       = {LIPIcs},
  volume       = {154},
  pages        = {37:1--37:18},
  publisher    = {Schloss Dagstuhl - Leibniz-Zentrum f{\"{u}}r Informatik},
  year         = {2020}
}

@book{RS94,
author = {G. Rozenberg and A. Salomaa},
title = {Cornerstones of Undecidability},
year = {1994},
publisher = {Prentice Hall}
}

@article{BRS06,
  author    = {D. Beauquier and
               A. M. Rabinovich and
               A. Slissenko},
  title     = {A Logic of Probability with Decidable Model Checking},
  journal   = {J. Log. Comput.},
  volume    = {16},
  number    = {4},
  year      = {2006}
}

@inproceedings{PiribauerB20,
  author       = {J. Piribauer and
                  C. Baier},
  title        = {On {S}kolem-Hardness and Saturation Points in {M}arkov Decision Processes},
  booktitle    = {{ICALP}},
  series       = {LIPIcs},
  volume       = {168},
  pages        = {138:1--138:17},
  publisher    = {Schloss Dagstuhl - Leibniz-Zentrum f{\"{u}}r Informatik},
  year         = {2020}
}

@article{BlondelT00,
  author    = {V. Blondel and
               J. Tsitsiklis},
  title     = {A survey of computational complexity results in systems
               and control},
  journal   = {Automatica},
  volume    = {36},
  number    = {9},
  year      = {2000},
  pages     = {1249-1274}
}

@inproceedings{FijalkowOPP019,
  author       = {N. Fijalkow and
                  J. Ouaknine and
                  A. Pouly and
                  J. Sousa Pinto and
                  J. Worrell},
  title        = {On the decidability of reachability in linear time-invariant systems},
  booktitle    = {{HSCC}},
  pages        = {77--86},
  publisher    = {{ACM}},
  year         = {2019}
}

@article{OWSiglog15,
  author    = {J. Ouaknine and
               J. Worrell},
  title     = {On linear recurrence sequences and loop termination},
  journal   = {{ACM} {SIGLOG} News},
  volume    = {2},
  number    = {2},
  pages     = {4--13},
  year      = {2015}
}

@article{KL86,
  author    = {R. Kannan and
               R. J. Lipton},
  title     = {Polynomial-Time Algorithm for the Orbit Problem},
  journal   = {JACM},
  volume    = {33},
  number    = {4},
  year      = {1986}
}

@article{CLZ00,
  author    = {J.-Y. Cai and
               R. J. Lipton and
               Y. Zalcstein},
  title     = {The Complexity of the {A B C} Problem},
  journal   = {SIAM J. Comput.},
  volume    = {29},
  number    = {6},
  year      = {2000}
}

@book{Berstel2010NoncommutativeRS,
  title={Noncommutative Rational Series with Applications},
  author={J. Berstel and C. Reutenauer},
publisher = {Cambridge University Press},
year={2010}
}

@article{cramer1936,
  title={On the order of magnitude of the difference between consecutive prime numbers},
  author={Cram{\'e}r, H.},
  journal={Acta Arith.},
  volume={2},
  pages={23--46},
  year={1936},
  publisher={Instytut Matematyczny Polskiej Akademii Nauk}
}

@inproceedings{LLNOP022,
  author       = {R. Lipton and
                  F. Luca and
                  J. Nieuwveld and
                  J. Ouaknine and
                  D. Purser and
                  J. Worrell},
  title        = {On the {S}kolem Problem and the {S}kolem Conjecture},
  booktitle    = {{LICS}},
  pages        = {5:1--5:9},
  publisher    = {{ACM}},
  year         = {2022}
}

@inproceedings{BLNOPW22,
  author       = {Y. Bilu and
                  F. Luca and
                  J. Nieuwveld and
                  J. Ouaknine and
                  D. Purser and
                  J. Worrell},
  title        = {Skolem Meets {S}chanuel},
  booktitle    = {{MFCS}},
  series       = {LIPIcs},
  volume       = {241},
  pages        = {20:1--20:15},
  publisher    = {Schloss Dagstuhl - Leibniz-Zentrum f{\"{u}}r Informatik},
  year         = {2022}
}

@inproceedings{LOW21,
  author       = {F. Luca and
                  J. Ouaknine and
                  J. Worrell},
  title        = {Universal {S}kolem Sets},
  booktitle    = {{LICS}},
  pages        = {1--6},
  publisher    = {{IEEE}},
  year         = {2021}
}

@inproceedings{LOW22,
  author       = {F. Luca and
                  J. Ouaknine and
                  J. Worrell},
  title        = {A Universal {S}kolem Set of Positive Lower Density},
  booktitle    = {{MFCS}},
  series       = {LIPIcs},
  volume       = {241},
  pages        = {73:1--73:12},
  publisher    = {Schloss Dagstuhl - Leibniz-Zentrum f{\"{u}}r Informatik},
  year         = {2022}
}

@article{nicely99,
  author       = {T. R. Nicely},
  title        = {New maximal prime gaps and first occurrences},
  journal      = {Math. Comput.},
  volume       = {68},
  number       = {227},
  pages        = {1311--1315},
  year         = {1999}
}

@article{ford18,
  title={Long gaps between primes},
  author={Ford, Kevin and Green, Ben and Konyagin, Sergei and Maynard, James and Tao, Terence},
  journal={Journal of the American Mathematical Society},
  volume={31},
  number={1},
  pages={65--105},
  year={2018}
}

\end{document}